\documentclass[12pt]{llncs}
\usepackage{etex}
\usepackage[a4paper, margin=10mm]{geometry}
\usepackage[english]{babel}
\usepackage{amsfonts}
\usepackage{inputenc}
\usepackage[ruled]{algorithm}
\usepackage{xcolor}
\usepackage{graphicx}
\usepackage{algpseudocode}
\usepackage{booktabs,multicol,multirow}
\usepackage{fancyhdr}
\usepackage{epic}
\usepackage{xspace}
\usepackage{amssymb,amsmath,verbatim}
\usepackage{multirow}
\usepackage{multicol}
\usepackage{adjustbox}
\usepackage{threeparttable}
\usepackage{float}

\usepackage{amsthm}
\usepackage{chngcntr}
\usepackage{url}
\usepackage[
n,
operators,
advantage,
sets,
adversary,
landau,
probability,
notions,	
logic,
ff,
mm,
primitives,
events,
complexity,
asymptotics,
keys]{cryptocode}
\usepackage{threeparttable, tablefootnote}

\def\ZZ{\mathbb{Z}}
\def\RR{\mathbb{R}}
\def\RR{\mathbb{R}}

\def\e{\bf{e}}
\def\b{\bf{b}}
\def\s{\bf{s}}
\def\m{\bf{m}}
\def\L{\Lambda}

\def\b{\bf{b}}
\def\s{\bf{s}}
\def\A{\bf{A}}
\def\H{\bf{H}}
\def\R{\bf{R}}
\def\G{\bf{G}}
\def\X{\bf{X}}

\def\B{\bf{B}}
\def\V{\bf{V}}
\def\x{\bf{x}}
\def\c{\bf{c}}
\def\IBPRE{\mathsf{IB\hbox{-}uPRE}}
\def\INDsCPA{\mathsf{IND\hbox{-}sID\hbox{-}CPA}}
\def\Exp{\mathsf{Exp}}
\def\Pr{\mathsf{Pr}}
\begin{document}
\title{Lattice-based Unidirectional IBPRE Secure in Standard Model}
\author{Priyanka Dutta\and Willy Susilo\and Dung Hoang Duong\and Joonsang Baek\and \\Partha Sarathi Roy}
\institute{Institute of Cybersecurity and Cryptology\\
School of Computing and Information Technology\\
University of Wollongong\\
	Northfields Avenue, Wollongong NSW 2522, Australia\\
\email{\{pdutta,wsusilo,hduong,baek,partha\}@uow.edu.au} 
}
\maketitle

\begin{abstract}
Proxy re-encryption (PRE) securely enables the re-encryption of ciphertexts from one key to another, without relying on trusted parties, i.e., it offers delegation of decryption rights. PRE allows a semi-trusted third party termed as a ``proxy" to securely divert encrypted files of user A (delegator) to user B (delegatee) without revealing any information about the underlying files to the proxy. To eliminate the necessity of having a costly certificate verification process, Green and Ateniese introduced an identity-based PRE (IB-PRE). The potential applicability of IB-PRE leads to intensive research from its first instantiation. Unfortunately, till today, there is no unidirectional IB-PRE secure in the standard model, which can withstand quantum attack. In this paper, we provide, for the first time, a concrete construction of unidirectional IB-PRE which is secure in standard model based on the hardness of learning with error problem. Our technique is to use the novel trapdoor delegation technique of Micciancio and Peikert. The way we use trapdoor delegation technique may prove useful for functionalities other than proxy re-encryption as well.
\end{abstract}

\section{Introduction}
Blaze, Bleumer and Strauss \cite{BBS98} introduced the concept of Proxy Re-encryption (PRE) towards an efficient solution that offers delegation of decryption rights without compromising privacy. PRE allows a semi-trusted third party, called a proxy, to securely divert encrypted files of one user (delegator) to another user (delegatee). The proxy, however, cannot learn the underlying message $m$, and thus both parties' privacy can be maintained. This primitive (and its variants) have various applications ranging from encrypted email forwarding \cite{BBS98}, securing distributed file systems \cite{AFGH07}, to digital rights management systems.
In addition application-driven purposes, various works have shown connections between re-encryption with other cryptographic primitives, such as program obfuscation \cite{HRV07,CCV12,CCLNX14} and fully-homomorphic encryption \cite{CLNTV15}. Thus studies along this line are both important and interesting for theory and practice.

PRE systems are classified as unidirectional and bidirectional based on the direction of delegation. It is worth mentioning that the unidirectional constructions are much desirable because bidirectional construction easily implementable using a unidirectional one. Though the concept of PRE was initiated in \cite{BBS98}, the first unidirectional PRE proposed by Ateniese et al. in \cite{AFGH07}, where following desired properties of a PRE are listed: {\em Non-interactivity} ({\em re-encryption key}, $rk _{A \rightarrow B}$, can be generated by $A$ alone using $B$'s public key; no trusted authority is needed); {\em Proxy transparency} (neither the delegator nor the delegatees are aware of the presence of a proxy); {\em Key optimality} (the size of $B$'s secret key remains constant, regardless of how many delegations he accepts); {\em Collusion resilience} (it is computationally infeasible  for the coalition of the proxy and user $B$ to compute $A$’s secret key); {\em Non-transitivity} (it should be hard for the proxy to re-delegate the decryption right, namely to compute $rk _{A \rightarrow C}$ from $rk _{A \rightarrow B}$, $rk _{B \rightarrow C}$). To achieve the aforementioned properties (partially) with improved security guarantee, there are elegant followup works which can be found in \cite{CH07,HRV07,LV08,CCV12,CCLNX14}. For quantum-safe version of PRE, Gentry \cite{Gentry09} mentioned the feasibility of unidirectional PRE through fully homomorphic encryption scheme (FHE). However, FHE costs huge computation. Xagawa proposed construction of PRE in \cite{Xagawa10}, but the construction lacks concrete security analysis. Further development of lattice-based PRE can be found in \cite{KIRSH14,CCLNX14,NX15,FL19}. 

Certificate management problem is a crucial issue in the PKI based schemes. This crucial issue was addressed by Green et al. \cite{GA07} in the area of PRE. For lattice-based construction, Singh et al. \cite{SRB13} proposed a bidirectional identity-based PRE. However, it is required to use secret key of both delegator and delegatee to generate re-encryption key, which lacks one of the fundamental properties of PRE. Further, they proposed unidirectional identity-based PRE \cite{SRB14u}, termed as {\em $\IBPRE$}, secure in the random oracle model. However, the size of re-encrypted ciphertext blows up than the original encrypted one. Moreover, the schemes encrypt the message bit by bit. Later, there are some further attempts to construct lattice-based identity-based PRE, which are flawed\footnote{In \cite{HJGS18}, authors claimed to proof IND-ID-CPA, but provide the proof for IND-CPA. In \cite{yin17}, authors assumed a universally known entity ({$\G$} matrix; see section \ref{1.lattices}) as a secret entity.} \cite{HJGS18,yin17}.

\noindent{\bf Our Contribution and Technique:} It is an interesting open research problem to construct post-quantum secure $\IBPRE$ in the standard model. In this paper, we resolve this daunting task by constructing a concrete scheme based on the hardness of {\em learning with error} (LWE) problem. The proposed construction is capable of encrypting multi-bit message and enjoy the properties like non-interactivity, proxy transparency, key optimality, non-transitivity along with other properties follow generically from IB-PRE. To construct the $\IBPRE$, we start with the construction of the identity-based encryption scheme by Agrawal et al. \cite{ABB10-EuroCrypt}. In non-interactive $\IBPRE$, it is required to construct re-encryption key by the delegator alone. One of the feasible ways to adopt the non-interactive feature is to provide a trapdoor to the delegator as a secret key. But, this technique is not supported by the design of \cite{ABB10-EuroCrypt}. In \cite{ABB10-EuroCrypt}, the trapdoor is the master secret key and the secret key of user is sampled by the master secret key. We first trace the design of selective IBE, where the secret key of a user is also a trapdoor, by using the trapdoor delegation technique of \cite{MP12}. Then extend the design to incorporate re-encryption feature based on the encryption scheme of \cite{MP12}. Here, the secret key of a user is a tuple of trapdoor, where one is used for decryption and another one is used for re-encryption key (ReKey) generation. ReKey is generated as in \cite{KIRSH14,FL19} with a trick to resists proxy to get any information regarding the underlying message of the corresponding re-encrypted ciphertext.. The underlying IBE of the proposed $\IBPRE$ may prove useful to design expressive cryptographic primitives other than identity-based proxy re-encryption as well.

\section{Preliminaries}
We denote the real numbers and the integers  by $\mathbb{R}, \mathbb{Z}$, respectively.
We denote column-vectors by lower-case bold letters (e.g. $\bf{b}$), so row-vectors are represented via transposition (e.g. ${\bf{b}}^t$). Matrices are denoted by upper-case bold letters and treat a matrix $ {\X} $ interchangeably with its ordered set $ \{{\bf{x}}_1,{\bf {x}}_2, \ldots\}$ of column vectors. We use ${\bf{I}}$ for the identity matrix and ${\bf{0}}$ for the zero matrix, where the dimension will be clear from context. We use $[ * | * ]$ to denote the concatenation of vectors or matrices.  A negligible function, denoted generically by $\negl$. We say that a probability is overwhelming if it is $1 - \negl$. The {\em statistical distance} between two distributions ${\bf{X}}$ and ${\bf{ Y}}$ over a countable domain $\Omega$ defined as $\frac 1{2} \sum_{w \in \Omega} | \Pr[{\bf{X}}=w]- \Pr[{\bf{Y}}=w]|.$ We say that a distribution over $\Omega$ is $\epsilon$-far if its statistical distance from the uniform distribution is at most $\epsilon$. Throughout the paper, $r = \omega (\sqrt{\log n})$ represents a fixed function which will be approximated by $\sqrt{ \ln(2n/\epsilon)/\pi}$.

\subsection{Lattices}
\label{1.lattices}
A $lattice ~\Lambda $ is a discrete additive subgroup of $\mathbb{R}^{m}$. Specially, a lattice $\Lambda$ in $\mathbb{R}^{m}$ with basis ${\B}=[{{\b}_1,\cdots,{\b}_n}]\in\mathbb{R}^{m\times n}$, where each ${\b}_i$ is written in column form, is defined as $\Lambda:=\left\{\sum_{i=1}^n{\b}_ix_i | x_i\in\ZZ~\forall i=1,\ldots,n \right\}\subseteq\mathbb{R}^m.$
We call $n$ the rank of $\L$ and if $n=m$ we say that $\L$ is a full rank lattice. The dual lattice $\Lambda ^{*}$  is the set of all vectors ${ \bf{y}} \in \mathbb{R}^{m} $ satisfying $\langle \bf{x},\bf{y} \rangle \in \mathbb{Z} $  for all vectors ${\bf{x}} \in \Lambda $. If $\B$ is a basis of an arbitrary lattice $\Lambda$, then $ {\B} ^{*} = {\B} ({\B}^t{\B})^{-1}$ is a basis for $\Lambda ^{*}$. For a full-rank lattice, ${\B} ^{*}= {\B}^{-t}$. 

In this paper, we mainly consider full rank lattices containing $q\ZZ^m$, called $q$-ary lattices, defined as the following, for a given matrix ${\A}\in\ZZ_q^{n\times m}$ and ${\bf{u}}\in\ZZ_q^n$: $\Lambda ^{\perp} ({\A}) := \left\{ {\bf{z}} \in \mathbb{Z}^{m} : {\A} {\bf{z}} =0\!\!\!\mod q \right\}$; $\Lambda ({\A }^t) =  \left\{ {\bf{z}} \in \mathbb{Z}^{m} : \exists~ {\bf{s}} \in \mathbb{Z}_q^{n}  ~s.t.~ {\bf{z}} = { \A}^t {\bf{s}}\!\!\!\mod q \right\}$; $\Lambda _{\bf{u}} ^{\perp} ({\A}) := \left\{{\bf{ z}} \in \mathbb{Z}^{m} : {\A} {\bf{z}} ={\bf{u}}\!\!\!\mod q \right\} = \Lambda ^{\perp} ({\A}) + {\bf x} ~for~ {\bf x} \in \Lambda ^{\perp} ({\A})$. Note that, $\Lambda ^{\perp} (\A)$ and $\Lambda ({\A}^t)$ are dual lattices, up to a $q$ scaling factor: $ q \Lambda ^{\perp} ({\A})^{*} = \Lambda ({\A} ^t)$, and vice-versa. Sometimes we consider the non-integral, $1$-$ary$  lattice $ \frac 1{q}\Lambda ({\A} ^t)= \Lambda ^{\perp} ({\A})^{*} \supseteq \mathbb{Z} ^{m}$.

\noindent{\bf Gaussian on Lattices:} Let $\L\subseteq\ZZ^m$ be a lattice. For a vector ${\bf c}\in{\RR}^m$ and a positive parameter $s\in\RR$, define: $\rho_{{\c}, s}({\x})=\exp\left(\pi\frac{\|{\x}-{\c}\|^2}{s^2}\right)\text{and} ~\rho_{{\c}, s}(\L)=\sum_{\x\in\L}\rho_{{\c}, s}(\x).$
The discrete Gaussian distribution over $\L$ with center $\c$ and parameter $\sigma$ is ${\cal{D}}_{{\L},{\c}, s}({\bf{y}})=\frac{\rho_{{\c}, s}(\bf{y})}{\rho_{{\c}, s}(\L)}, \forall {\bf y}\in\L$.

\noindent{\bf Hard problems on Lattices:} There are two lattice-based one-way functions associated with matrix ${\A }\in \mathbb{Z}_q ^{n\times m}$ for $m=poly(n)$:
  \begin{itemize}
 \item $g_{\A}({\bf{e}},{\bf{s}}) = {\bf{s}}^t {\A} + {\bf{e}}^t \mod q$ for $ {\bf{s}} \in \mathbb{Z}_q^{n} $ and a Gaussian ${\bf{e}} \in \mathbb{Z}^{m}$ and
   $ f_{\A}({\bf{x}}) = {\A \bf{x}} \mod q$, for ${\bf{x}} \in \mathbb{Z}^{m};$
 
  \item The Learning With Errors (LWE) problem was introduced in \cite{Regev05}. The problem to invert $g_{\A}({\bf{e}},{\bf{s}})$, where $ {\bf{e}} \leftarrow D_{{\mathbb{Z}^{m}}, \alpha q}$ is known as search-LWE$_{q, n, m, \alpha} $ problem and is as hard as quantumly solving Shortest Independent Vector Problem (SIVP) on $n$-dimensional lattices. The decisional- LWE$_{q, n, m, \alpha}$ problem asks to distinguish the output of $g_{\A}$ from uniform.

  \item The Small Integer Solution (SIS) problem was first suggested to be hard on average by Ajtai \cite{ajtai96} and then formalized by Micciancio and Regev \cite{MR04}. Finding a non-zero short preimage $\bf{x}'$ such that $f_{\A}(\bf{x}')= \bf{0}$, with $\norm{\bf{x}'} \leq \beta$, is an instantiation of the SIS$_{q, n, m, \beta}$ problem. It is known to be as hard as certain worst-case problems (e.g. SIVP) in standard lattices \cite{Ajtai99,MR04,GPV08,MP13}.
 \end{itemize}

 \noindent{\bf Trapdoors for Lattices:} Here, we briefly describe the main results of \cite{MP12} and it's generalized version from \cite{KIRSH14}: the definition of $\G$-trapdoor, the algorithms $\bf{Invert^{\mathcal{O}}}$, $\bf{Sample}^{\mathcal{O}}$ and ${\bf{DelTrap}} ^{\mathcal{O}}$.
 
 A $\G$-trapdoor is a transformation (represented by a matrix $\R$) from a public matrix $\A$ to a special matrix $\G$ which is called as gadget matrix. The formal definitions as follows:
 
 \begin{definition} [\cite{MP12}]
 Let ${\A }\in \mathbb{Z}_q^{n\times m}$ and ${\G }\in \mathbb{Z}_q^{n\times w}$ be matrices with $m \geq w \geq n$. A $\G$-trapdoor for $\A$ is a matrix ${\R} \in \mathbb{Z}^{(m-w)\times w}$ such that $\A  \left[ \begin{array}{c} \R \\ \bf{I} \end{array}\right ] = \H \G$, for some invertible matrix ${\H} \in \mathbb{Z}_q^{n \times n}$. We refer to $\H$ as the tag or label of the trapdoor. 
 \end{definition}

 \begin{definition}[\cite{KIRSH14}]
 The generalized version of a $\G$-trapdoor $:$\\
 Let ${\A }= \left[ \begin{array}{c|c|c|c} {\A}_0 & {\A}_1 & \cdots & {\A}_{k-1} \end{array}\right ] \in \mathbb{Z}_q^{n \times m}$ for $ k \geq 2$, and ${\A}_0 \in \mathbb{Z}_q^{n \times \bar{m}}, {\A}_1, \ldots, {\A}_{k-1}  \in \mathbb{Z}_q^{n \times w}$ with $ \bar{m} \geq w \geq n $ and $ m= \bar{m} +(k-1) \cdot w$ $(typically,~ w= n \lceil \log q \rceil)$. A $\G$-trapdoor for $\A$ is a sequence of matrices ${\R}= \left[\begin{array}{c|c|c|c} {\R}_1 &{ \R}_2 & \cdots &{ \R}_{k-1} \end{array}\right]\in \mathbb{Z}_q^{\bar{m} \times (k-1)w}$ such that $:$ 

$$\left[ \begin{array}{c|c|c|c} {\A}_0 & {\A}_1 & \cdots & {\A}_{k-1} \end{array}\right ] \left[ \begin{array}{cccc} {\R}_1 & {\R}_2 & \cdots & {\R}_{k-1} \\ {\bf{I}} & {\bf{0}} & \cdots & {\bf{0}} \\ \vdots & \vdots & \ddots & \vdots \\ {\bf{0}}& {\bf{0}} & \cdots & {\bf{I}}\end{array}\right ]  = \left[ \begin{array}{c|c|c|c} {{\H}_1}{\G}&{ {\H}_2}{\G} & \cdots & {{\H}_{k-1}}{\G} \end{array}\right ],$$

for invertible matrices ${{\H}_i} \in \mathbb{Z}_q^{n \times n}$ and a fixed ${\G}\in \mathbb{Z}_q^{n\times w}$.
 \end{definition}
 
 \paragraph{${\bf{Invert}^{\mathcal{O}}} ({\R}, {\A}, {\b}, {\H}_{i})$}\cite{KIRSH14}: On input a vector ${\b}^t = {\s}^t {\A} +{\e}^t$, a matrix\\
 $\A = \left[ \begin{array}{c|c|c|c} {{\A}_0} & -{{\A}_0}{\R}_1+{\H_1\G } & \cdots & -{{\A}_0}{\R}_{k-1}+{\H}_{k-1}\G \end{array}\right]$  and its corresponding $\G$-trapdoor ${\R}= \left[\begin{array}{c|c|c|c} {\R}_1 &{ \R}_2 & \cdots &{ \R}_{k-1} \end{array}\right]$ with invertible tag ${\H}_i$, the algorithm first computes 
 $${\b}' = {\b}^t\left[ \begin{array}{cccc} {\R}_1 & {\R}_2 & \cdots & {\R}_{k-1} \\ {\bf{I}} & {\bf{0}} & \cdots & {\bf{0}} \\ \vdots & \vdots & \ddots & \vdots \\ {\bf{0}}& {\bf{0}} & \cdots & {\bf{I}}\end{array}\right ]$$ and then run the inverting oracle  $\mathcal{O}(\b')$ for $\G$ to get $(\s',\e')$. The algorithm outputs ${\s} = {\H}_i^{-1} \s' $ and ${\e} = {\b} - {\A}^t {\s}$. Note that, ${\bf{Invert}^{\mathcal{O}}}$ produces correct output if ${\e} \in \mathcal{P}_{1/2}(q \cdot \mathbf{B}^{-t})$, where $\mathbf{B}$ is a basis of $\L^{\perp} ({\G})$; cf. \cite[Theorem 5.4]{MP12}.

\paragraph{${\bf{Sample}}^{\mathcal{O}} ({\R,\A,{\H}}, {\bf{u}}, s)$} \cite{MP12}: On input  $({\R,\A',\H}, {\bf{u}}, s)$, the algorithm construct $ \A = \left[ \begin{array}{c|c} {\A'} & -{\A}'\R+\H\G \end{array}\right]$, where $\R$ is the $\G$-trapdoor for matrix $\A$ with invertible tag $\H$ and ${\bf{u}} \in \mathbb{Z}_q^{n}$.
The algorithm outputs, using an oracle $\mathcal{O}$ for Gaussian sampling over a desired coset $\Lambda_{\bf{v}} ^{\perp} (\G)$, a vector drawn from a distribution within negligible statistical distance of $D _{\Lambda_{\bf{u}} ^{\perp} ({\A}), s}$. To sample a Gaussian vector ${\bf x}\in \mathbb{Z}_q^{m}$ for ${\A }= \left[ \begin{array}{c|c|c|c} {\A}_0 & {\A}_1 & \cdots & {\A}_{k-1} \end{array}\right ] \in \mathbb{Z}_q^{n \times m}$ with the generalized trapdoor  ${\R}= \left[\begin{array}{c|c|c|c} {\R}_1 &{ \R}_2 & \cdots &{ \R}_{k-1} \end{array}\right]$ and $k-1$ invertible ${\H}_i$'s given a coset ${\bf{u}} \in \mathbb{Z}_q^{n}$, use generalized version of ${\bf{Sample}}^{\mathcal{O}}$ from \cite{KIRSH14}. 


 \paragraph{${\bf{DelTrap}}^{\mathcal{O}}$}$({\A}' = \left[ \begin{array}{c|c} {\A} & {\A}_1 \end{array}\right], {\R},  {\H}', s)$\cite{MP12}: On input an oracle $\mathcal{O}$ for discrete Gaussian sampling over cosets of $ \Lambda = \Lambda^{\perp}(\A)$ with parameter $s$, an extended matrix $\A'$ of $\A$, an invertible matrix $\H'$, the algorithm will sample (using $\mathcal{O}$) each column of $\R'$ independently from a discrete Gaussian with parameter $s$ over the appropriate coset of $\Lambda ^{\perp}(\A)$, so that ${\A}{\R' }= {\H'} {\G} - {\A}_{1}$. The algorithm outputs a trapdoor $\R'$ for $\A'$ with tag ${\H}'$.

\subsection{ Identity-Based Unidirectional Proxy Re-Encryption}
\begin{definition} [Identity-Based Unidirectional Proxy ReEncryption (IB-uPRE) \cite{GA07}]
A unidirectional Identity-Based Proxy Re-Encryption $(\IBPRE)$ scheme is a tuple of algorithms  $({\bf{Set Up, Extract, ReKeyGen, Enc, ReEnc, Dec}}):$

\begin{itemize}

\item$(PP,msk) \longleftarrow {\bf{SetUp}}(1^n):$ On input the security parameter $1^{n}$, the $\bf{set up}$ algorithm outputs $PP, msk$.

\item $ sk_{id} \longleftarrow {\bf{Extract}}(PP, msk, id):$ On input an identity $id$, public parameter PP, master secret key, output the secret key $sk_{id}$ for $id$.

\item $rk _{i \rightarrow j} \longleftarrow\mathbf{ReKeyGen}(PP, sk_{id_{i}}, {id_i}, {id_j}):$ On input a public parameter $PP$, secret key $sk_{id_{i}}$ of a delegator $i$, and $id_i, id_j$, output a unidirectional re-encryption key $rk _{i \rightarrow j} $.

\item $ct \longleftarrow {\bf{Enc}}(PP, id, m):$ On input an identity $id$, public parameter $PP$ and a plaintext $m \in \mathcal{M}$, output a ciphertext $ct$ under the specified identity $id$.

\item $ct' \longleftarrow {\bf{ReEnc}}(PP,rk _{i \rightarrow j}, ct):$ On input a ciphertext $ct$ under the identity $i$  and a re-encryption key $rk _{i \rightarrow j}$, output a ciphertext $ct'$ under the identity $j$.

\item $m \longleftarrow {\bf{Dec}}(PP,sk_{id_{i}}, ct):$ On input the ciphertext $ct$ under the identity $i$ and secret key $sk_{id_{i}}$ of $i$, the algorithm outputs a plaintext $m$ or the error symbol $\bot$.

\end{itemize} 
\end{definition}

An Identity-Based Proxy Re-Encryption scheme is called  single-hop if a ciphertext can be re-encrypted only once. In a multi-hop setting   proxy can apply further re-encryptions to already re-encrypted ciphertext. 

\begin{definition} [Single-hop IB-uPRE Correctness]
A single-hop unidirectional Identity-Based Proxy Re-Encryption scheme $({\bf{Set Up, Extract, ReKeyGen, Enc, ReEnc, Dec}})$ decrypts correctly for the plaintext space $\mathcal{M}$ if $:$
\begin{itemize}
\item For all $ sk_{id}$ output by ${\bf{Extract}}$ under $id$ and for all $m \in \mathcal{M}$,\\ it holds that $ {\bf{Dec}}(PP, sk_{id}, {\bf{Enc}}(PP, id, m))  = m$.

\item For any re-encryption key $rk _{i \rightarrow j}$ output by $\mathbf{ReKeyGen} (PP, sk_{id_{i}}, {id_i}, {id_j})$ and any $ct={\bf{Enc}}(PP, id_i,m)$ it holds that  ${\bf{Dec}}(PP, sk_{id_{j}}, {\bf{ReEnc}}(PP,rk _{i \rightarrow j}, ct)) = m$.

\end{itemize}
\end{definition}

\paragraph{{\bf Security Game of Unidirectional Selective Identity-Based Proxy Re-Encryption Scheme against Chosen Plaintext Attack (IND-sID-CPA)}}: $\\$
To describe the security model we first classify all of the users into honest $(HU)$ and corrupted $(CU)$. In the honest case an adversary does not know secret key, whereas for a corrupted user the adversary has secret key.
Let $\mathcal{A}$ be the PPT adversary and $\Pi = ({\bf{Set Up, Extract, ReKeyGen}},$ ${\bf{Enc, ReEnc, Dec}})$ be an $\IBPRE$ scheme with a plaintext space $\mathcal{M}$ and a ciphertext space $\mathcal{C}$. Let $id^{*}(\in HU)$ be the target user. Security game is defined according to the following game $\Exp_{\mathcal{A}}^ {\INDsCPA} (1^{n}):$

\begin{enumerate}
\item $\bf{Set Up}$: The challenger runs ${\bf{Set Up}}(1^{n})$ to get ($PP, msk)$ and give $PP$ to $\mathcal{A}$. 
\item{{\bf Phase 1:}} The adversary $\mathcal{A}$ may make quires polynomially many times in any order to the following oracles:
\begin{itemize}
\item $\mathcal{O}^{\bf{Extract}}$: an oracle that on input $id \in CU$, output $sk_{id}$; Otherwise, output $\bot$.

\item $\mathcal{O}^{\bf{ReKeyGen}}$: an oracle that on input the identities of $i$-th and $j$-th users: if $id_i \in HU \setminus \{id^{*}\}$, $id_j \in HU$ or $id_i, id_j \in CU$ or $id_i \in CU, id_j \in HU$, output $rk _{i \rightarrow j}$; otherwise, output $\bot$.

\item $\mathcal{O}^{\bf{ReEnc}}$: an oracle that on input the identities of $i, j$-th users and ciphertext of $i$-th user: if $id_i, id_j \in HU$ or $id_i, id_j \in CU$ or $id_i \in CU, id_j \in HU $output re-encrypted ciphertext; otherwise, output $\bot$.
\end{itemize}

\item $\bf{Challenge}$: $\mathcal{A}$ outputs two messages $m_0, m_1 \in \mathcal{M}$ and is given a challenge ciphertext $ct_b \longleftarrow {\bf{Enc}}(PP, id ^{*},m_b)$ for either $b=0$ or $b=1$. 

\item{{\bf Phase 2:}} After receiving the challenge ciphertext, $\mathcal{A}$ continues to have access to the $\mathcal{O}^{\bf{Extract}}$, $\mathcal{O}^{\bf{ReKeyGen}}$ and $\mathcal{O}^{\bf{ReEnc}}$ oracle as in {\bf Phase 1}.

\item $\mathcal{O}^{\bf{Decision}}$: On input $b'$ from adversary  $\mathcal{A}$, this oracle outputs $1$ if $b = b'$ and $0$ otherwise. 
\end{enumerate}

\noindent The advantage of an adversary in the above experiment $\Exp_{\mathcal{A}}^ {\INDsCPA}$$(1^{n})$ is defined as $|\Pr[b' = b]-\frac1{2} |$.

\begin{definition}
\label{def:security}
An $\IBPRE$ scheme is $\INDsCPA$ secure if all PPT adversaries $\mathcal{A}$ have at most a negligible advantage in experiment  $\Exp_{\mathcal{A}}^ {\INDsCPA}(1^{n})$.
\end{definition}

\begin{remark}
In \cite{GA07}, $\bf{ReKeyGen}$ query is allowed from $id^{*}$ to $HU$ to make the $\IBPRE$ collusion resilient (coalition of malicious proxy and delegetee to compute delegator’s secret key). Here, we have blocked $\bf{ReKeyGen}$ query from $id^{*}$ to $HU$ and the proposed $\IBPRE$ scheme is not claimed to be collusion resilient.
\end{remark}

\section{Single-hop Identity-Based Unidirectional Proxy Re-Encryption Scheme (IB-uPRE)}
\subsection{Construction of Single-hop IB-uPRE}\label{sec:construction}
In this section, we present our construction of single-hop $\IBPRE$. We set the parameters as the following.
\begin{itemize}
\item $\mathbf{G}\in\mathbb{Z}_q^{n \times nk}$ is a gadget matrix for large enough prime power $ q = p^ {e}  = poly (n)$ and $k = O (\log q ) = O (\log n) $, so there are efficient algorithms to invert $ g_\mathbf{G} $ and to sample for $ f_\mathbf{G}$. 

\item $\bar{m} = O(nk) $ and  the Gaussian $ \mathcal{D} = D_{\mathbb{Z}, r}^{\bar{m}\times nk }$, so that $( \bar{\A}, \bar{\A} \R)$ is negl(n)-far from uniform for $\bar{\A}$.

 \item the LWE error rate $\alpha $  for $\IBPRE$ should satisfy $1 / \alpha = O (nk)^{3} \cdot r^3$. 
\end{itemize}

To start out, we first recall encoding techniques from \cite{MP12,ABB10-EuroCrypt}.
\begin{itemize}
\item{\bf Message Encoding:} In the proposed construction, message space is $\mathcal{M} = \{0, 1\}^ {nk}$. $\mathcal{M}$ map bijectively to the cosets of $\Lambda / 2 \Lambda $ for $ \Lambda = \Lambda(\mathbf{G}^{t})$ by some function $ encode$  that is efficient to evaluate and invert. In particular, letting $ \mathbf{E} \in\mathbb{Z}^{nk \times nk} $ be any basis of $\Lambda$, we can map ${\bf{m}}\in\{0, 1\}^ {nk}$ to $encode({\bf{m}}) $= $\mathbf{E}{\bf{m}} \in \mathbb{Z}^{nk}$ \cite{MP12}.

\item{\bf Encoding of Identity:} In the following construction, we use {\em full-rank difference} map (FRD) as in \cite{ABB10-EuroCrypt}. FRD: $\mathbb{Z}_q^{n} \rightarrow \mathbb{Z}_q^{n\times n}$; $id \mapsto {\H}_{id}$. We assume identities are non-zero elements in $ \mathbb{Z}_q^{n} $. The set of identities can be expanded to $\{0,1\}^{*}$ by hashing identities into $ \mathbb{Z}_q^{n} $ using a collision resistant hash. FRD satisfies the following properties: 1. $\forall ~ distinct ~id_1, id_2 \in  \mathbb{Z}_q^{n}$, the matrix ${\H}_{id_1} - {\H}_{id_2} \in \mathbb{Z}_q^{n\times n}$ is full rank; 2. $\forall ~ id \in  \mathbb{Z}_q^{n} \setminus \{{\bf 0}\}$, the matrix ${\H}_{id} \in \mathbb{Z}_q^{n\times n}$ is full rank; 3. FRD is computable in polynomial time (in $n \log q$).
\end{itemize}

The proposed $\IBPRE$ consists of the following algorithms: 

\paragraph{$\mathbf{SetUp}(1^n):$} On input a security parameter $n$, do:
\begin{enumerate}
\item Choose $\bar{\A} \leftarrow\mathbb{Z}_q^{n \times \bar{m}} $, $\R \leftarrow \mathcal{D}$, and set $\bar{\A}' =  -\bar{\A}\R $ $\in \mathbb{Z}_q^{n \times nk}$.
\item Choose four invertible matrices ${\H}_1, {\H}_2, {\H}_3, {\H}_4$ at uniformly random from $\mathbb{Z}_q^{n\times n}$. 
\item Choose two random matrices ${\A}_1, {\A}_2$ from $ \mathbb{Z}_q^{n \times nk}$.
\item Output $PP= (\bar{\A}, \bar{\A}', {\A}_1, {\A}_2, {\H}_1, {\H}_2, {\H}_3, {\H}_4, \G)$ and the master secret key is $msk = \R$.
\end{enumerate}

\paragraph{$\mathbf{Extract}( PP, msk, id):$} On input a public parameter $PP$, master secret key $msk$ and the identity of $i$-th user $id_i$, do:
\begin{enumerate}
\item Construct ${\tilde{\A}}_{i} = \left [ \begin{array}{c  |  r} \bar{\A}  &  \bar{\A}' +$${\H}_{id{_i}}$$\G$$\end{array} \right ]  =\left [ \begin{array}{c  |  r} \bar{\A}  &  -\bar{\A} \R +$${\H}_{id{_i}}$$\G$$\end{array} \right ]  $ $\in\mathbb{Z}_q^{n \times m}$, where $ m = \bar{m} + nk $. So, $\R$ is a trapdoor of ${ \tilde{\A}}_{i}$ with tag ${\H}_{id{_i}}$.

\item 
\begin{itemize}
\item Construct ${\A}_{i1} = {\A}_1+ {\H}_3 {\H}_{id_i} {\G}  \in \mathbb{Z}_q^{n \times nk}$ and set ${\A}_{i1}' = \left [ \begin{array}{c  |  c} {\tilde{\A}}_{i} &  {\A}_{i1} \end{array} \right ] \in\mathbb{Z}_q^{n \times (m+ nk)}$.

\item Call  the algorithm $\mathbf{DelTrap}^{\mathcal{O}}({{\A}_{i1}', \R}, {\H}_1, s)$ to get a trapdoor ${\R}_{i1}$ $\in\mathbb{Z}^{m \times nk} $ for ${\A}_{i1}' $ with tag ${\H}_1\in \mathbb{Z}_q^{n \times n} $, where $ s\geq \eta_{\epsilon}(\Lambda ^{\bot} ({\tilde{\A}}_{i}))$, so that  ${\tilde{\A}}_{i}{\R}_{i1} = {\H}_1{\G} - {\A}_{i1}$.
\end{itemize}

\item
\begin{itemize}
\item Construct ${\A}_{i2} = {\A}_2+ {\H}_4 {\H}_{id_i} {\G}  \in \mathbb{Z}_q^{n \times nk}$ and set ${\A}_{i2}' = \left [ \begin{array}{c  |  c}{ \tilde{\A}}_{i} &  {\A}_{i2} \end{array} \right ] \in\mathbb{Z}_q^{n \times (m+ nk)}$.

\item Call  the algorithm $\mathbf{DelTrap}^{\mathcal{O}}({\A}_{i2}', {\R}, {\H}_2, s)$ to get a trapdoor ${\R}_{i2} \in\mathbb{Z}^{m \times nk}$ for ${\A}_{i2}'$ with tag ${\H}_2 \in\mathbb{Z}_q^{n \times n}$, so that ${\tilde{\A}}_{i}{\R}_{i2} = {\H}_2 {\G} - {\A}_{i2}$.
\end{itemize}



Output the secret key as $sk_{id_{i}}=  \left [ \begin{array}{c  |  c  } {\R}_{i1} &  {\R}_{i2} \end{array} \right ] \in \mathbb{Z}^{m \times 2nk}$. 
Notice that, 
$$\left [ \begin{array}{c  |  c  | c}  {\tilde{\A}}_{i} &  {\A}_{i1}  &  {\A}_{i2}\end{array} \right ] \left [ \begin{array}{c c  } {\R}_{i1} &  {\R}_{i2} \\ \mathbf{I} & \mathbf{0}\\ \mathbf{0} & \mathbf{I} \end{array} \right ] = \left [ \begin{array}{c  |  c  } {\H}_1{\G} & {\H}_2{\G}\end{array} \right ].$$
\end{enumerate}

\paragraph{$\mathbf{Enc}(PP, id_{i}, {\m} \in \{0,1\}^{nk}):$} On input a public parameter $PP$, the identity of $i$-th user $id_i$ and message ${\m} \in \{0,1\}^{nk}$, do:
\begin{enumerate} 
\item Construct ${\tilde{\A}}_{i} =\left [ \begin{array}{c  |  r} \bar{\A}  &  -\bar{\A} \R +$${\H}_{id{_i}}$$\G$$\end{array} \right ]  $ $\in\mathbb{Z}_q^{n \times m}$. 
\item Construct  ${\A}_{i1}, {\A}_{i2}$ for $id_{i}$ same as in $\mathbf{Extract}$ algorithm and set  $ {\A}_i$ = $\left [ \begin{array}{c  |  c  | c} { \tilde{\A}}_{i} & {\A}_{i1} & {\A}_{i2}  \end{array} \right ]$.
 \item Choose a uniformly random  ${\bf s} \leftarrow \mathbb{Z}_q^{n}$.
 \item Sample error vectors $ \bar{{\e}_{0}}  \leftarrow  D_{\mathbb{Z},\alpha q}^{\bar{m}}$   and    $ {{\e}_{0}'},   {\e}_{1}, {\e}_{2} \leftarrow D_{\mathbb{Z},s'}^{nk}$, where $s'^{2} = (\| {\bar \e}_{0} \| ^ {2} + \bar {m} (\alpha q)^{2}) r^2$. Let the error vector $ {\e} = ({\e}_{0},  {\e}_{1},  {\e}_{2}) \in{\mathbb{Z} ^{\bar{m}+{nk}} }\times{\mathbb{Z}^{nk}}\times{\mathbb{Z} ^{nk}}$, where ${\e}_{0} = (\bar{\e}_{0} , {\e}_{0}') \in{\mathbb{Z} ^{\bar{m}}}\times{\mathbb{Z}^{nk}}$.   
\item Compute $ {\b} ^t = ( {\b}_{0},  {\b}_{1},  {\b}_{2}) = 2 ({\bf s}^t {\A}_i \mod q) + {\e}^t + (\mathbf{0},  \mathbf{0}, encode({\m})^t) \mod 2q$, where the first zero vector has dimension $\bar{m} +nk$, the second has dimension $nk$ and ${\b}_{0} = ({\bar\b}_{0}, {\b}_{0}')$.
\item Output the ciphertext $ct = {\b} \in\mathbb{Z}_{2q}^{\bar{m} + 3nk}$.
\end{enumerate}

\paragraph{$\mathbf{Dec}(PP, sk_{id_{i}}, ct):$} On input a public parameter $PP$, the secret key of $i$-th user $sk_{id_{i}}$ and ciphertext $ct$, do:
\begin{enumerate}
\item If $ct$ has invalid form or ${\H}_{id{_i}} = \mathbf{0}$, output $\bot$. Otherwise,
\begin{itemize} 
\item Construct ${\tilde{\A}}_{i} =\left [ \begin{array}{c  |  r} \bar{\A}  &  -\bar{\A} \R +$${\H}_{id{_i}}$$\G$$\end{array} \right ]  $ $\in\mathbb{Z}_q^{n \times m}$. 
\item Construct  ${\A}_{i1}, {\A}_{i2}$ for $id_{i}$ same as in $\mathbf{Extract}$ algorithm and set ${\A}_i$ = $\left [ \begin{array}{c  |  c  | c} { \tilde{\A}}_{i} & {\A}_{i1} & {\A}_{i2}  \end{array} \right ]$.
\end{itemize}
\item Call $\mathbf{Invert^{\mathcal{O}}}(\left [ \begin{array}{c  |  c  } {\R}_{i1} &  {\R}_{i2} \end{array} \right ],  {\A}_i, {\b}, {\H}_2)$  to get values ${\bf z} \in\mathbb{Z}_q^{n}$  and ${\e} = ({\e}_{0}, {\e}_{1},  {\e}_{2}) \in{\mathbb{Z} ^{\bar{m}+{nk}} }\times{\mathbb{Z}^{nk}}\times{\mathbb{Z} ^{nk}}$,  where ${\e}_{0} = (\bar{\e}_{0}, {\e}_{0}') \in{\mathbb{Z} ^{\bar{m}}}\times{\mathbb{Z}^{nk}}$ for which 
${\b}^t = {\bf{z}}^t {\A}_i + {\e}^t  \mod q$. If the call to $\mathbf{Invert}$ fails for any reason, output $\bot$.

\item If $\|{\bar\e}_{0}\| \geq \alpha q \sqrt{\bar{m}}$ or  $\|{\e}_{0}'\| \geq \alpha q \sqrt{2\bar{m}nk} \cdot r$ or $\|{\e}_{j}\| \geq \alpha q \sqrt{2\bar{m}nk} \cdot r$ for $j = 1,2 $, output $\bot$.

\item Let ${\V}= {\b}-{\e} \mod 2q$, parsed as ${\V}=({\V}_{0}, {\V}_{1}, {\V}_{2}) \in{\mathbb{Z} _{2q}^{\bar{m}+{nk}} }\times{\mathbb{Z}_{2q}^{nk}}\times{\mathbb{Z} _{2q}^{nk}}$, where ${\V}_{0} = ({\overline \V}_{0} , {\V}_{0}')\in{\mathbb{Z} _{2q}^{\bar{m}}}\times{\mathbb{Z}_{2q}^{nk}}$. If ${\overline \V}_0 \notin 2\Lambda(\bar{{\A}}^{t})$, output $\bot$.

\item Output $encode^{-1}$(${\V}^t \left [ \begin{array}{c c  } {\R}_{i1} &  {\R}_{i2} \\ \mathbf{I} & \mathbf{0}\\ \mathbf{0} & \mathbf{I} \end{array} \right ] \mod 2q)  \in \{0,1\}^{nk}$  if it exists, otherwise output $\bot$.

\end{enumerate}

\paragraph{$\mathbf{ReKeyGen}(PP, sk_{id_{i}}, id_i, id_j):$} On input a public parameter $PP$, the secret key of $i$-th user $sk_{id_{i}}$ and identity of $j$-th user $id_j$, do:
\begin{enumerate}
\item Construct $ {\A}_i$ = $\left [ \begin{array}{c  |  c  | c} { \tilde{\A}}_{i} & {\A}_{i1} & {\A}_{i2}  \end{array} \right ]$, where ${\tilde{\A}}_{i} =\left [ \begin{array}{c  |  r} \bar{\A}  &  \bar{\A}' +$${\H}_{id{_i}}$$\G$$\end{array} \right ]$ and ${\A}_{i1}, {\A}_{i2}$ are same as in $\mathbf{Extract}$ algorithm .
\item Construct  $ {\A}_j$ = $\left [ \begin{array}{c  |  c  | c} { \tilde{\A}}_{j} & {\A}_{j1} & {\A}_{j2}  \end{array} \right ]$, where  ${\tilde{\A}}_{j} =\left [ \begin{array}{c  |  r} \bar{\A}  &  \bar{\A}' +$${\H}_{id{_j}}$$\G$$\end{array} \right ]$ and ${\A}_{j1}, {\A}_{j2}$ are same as in $\mathbf{Extract}$ algorithm .

\item Using $\mathbf{Sample}^{\mathcal{O}}$ with trapdoor ${\R}_{i1}$(from the secret key of $i$th user ), with tag ${\H}_1$, we sample  from the cosets which are formed with the column of the matrix $ \bar{\A}' +$${\H}_{id{_j}}\G$. After sampling $nk $ times we get an $(\bar{m}+2nk)\times nk$ matrix and parse it as three matrices ${\X}_{00}$  $ \in\mathbb{Z}^{\bar{m} \times nk}$, ${\X}_{10}$  $ \in\mathbb{Z}^{nk \times nk}$ and ${\X}_{20}$  $ \in\mathbb{Z}^{nk \times nk}$ matrices with Gaussian entries of parameter $s$.
So, $$ \left [ \begin{array}{c  |  c }  {\tilde{\A}}_{i} &  -{\tilde{\A}}_{i} {\R}_{i1} +{\H}_1\G \end{array} \right ] \left [ \begin{array}{c } {\X}_{00} \\ {\X}_{10} \\{\X}_{20} \end{array} \right ] = \bar{\A}' +{\H}_{id{_j}}{\G}, ~~i.e. \left [ \begin{array}{c  |  c }  {\tilde{\A}}_{i} & {\A}_{i1}  \end{array} \right ] \left [ \begin{array}{c } {\X}_{00} \\ {\X}_{10} \\{\X}_{20}\end{array} \right ] =  \bar{\A}' +{\H}_{id{_j}}\G.$$
   
 \item Continue sampling for the cosets obtained from the columns of the matrix ${\A}_{j1}$ from $A_j$. 
 This time, we increase the Gaussian parameter of the resulting sampled matrix up to $s \sqrt{{\bar m}/2}$: 
$$ \left [ \begin{array}{c  |  c }  {\tilde{\A}}_{i} &  -{\tilde{\A}}_{i} {\R}_{i1} +{\H}_1 \G\end{array} \right ] \left [ \begin{array}{c } {\X}_{01} \\ {\X}_{11}\\{\X}_{21} \end{array} \right ] =  \begin{array}{c }  {\A}_{j1} \end{array}, ~~i.e.\left [ \begin{array}{c  |  c }  {\tilde{\A}}_{i} & {\A}_{i1}  \end{array} \right ] \left [ \begin{array}{c } {\X}_{01} \\ {\X}_{11}\\ {\X}_{21}\end{array} \right ] =  \begin{array}{c }  {\A}_{j1}\end{array}.$$

 For the last sampling, to get a correct re-encryption, we will use the cosets which are formed with the column of the matrix $\begin{array}{c } {\A}_{j2} + {\tilde{\A}}_{i} {{\R}_{i2} } -{\H}_2{\G } \end {array}$:
$$\left [ \begin{array}{c  |  c }  {\tilde{\A}}_{i}& -{\tilde{\A}}_{i}{\R}_{i1} + {\H}_1\G\end{array} \right ] \left [ \begin{array}{c } {\X}_{02} \\ {\X}_{12}\\{\X}_{22} \end{array} \right ] =  \begin{array}{c } {\A}_{j2}+{\tilde{\A}}_{i} {{\R}_{i2}} - {\H}_2{\G} \end{array},$$ 

 where ${\X}_{01}, {\X}_{02} \in \mathbb{Z}^{\bar{m}\times nk}$, ${\X}_{11}, {\X}_{12}, {\X}_{21}, {\X}_{22} \in \mathbb{Z}^{nk \times nk}$ with entries distributed as Gaussian with parameter $s\sqrt{\bar{m}}$.
 
 \item Output re-encryption key $rk_{i\rightarrow j}= \left[\begin{array}{cccc} \mathbf{I}_{\bar{m} \times \bar{m}} & {\X}_{00} & {\X}_{01} & {\X}_{02} \\  \mathbf{0}  & {\X}_{10}& {\X}_{11}& {\X}_{12}\\ \mathbf{0} & {\X}_{20} & {\X}_{21} & {\X}_{22}\\ \mathbf{0} & \mathbf{0} &  \mathbf{0} & \mathbf{I}_{nk \times nk}\end{array}\right] \in \mathbb{Z}^{(m+2nk)\times (m+2nk)}$,
  
which satisfies: ${\A}_i   \cdot rk_{i\rightarrow j} ={\A}_j$.


\end{enumerate}

\paragraph{$\mathbf{ReEnc}( rk_{i\rightarrow j}, ct):$} On input $rk_{i\rightarrow j}$ and $i$-th user's ciphertext $ct$, Compute:\\
 ${{\b}'}^t = {\b}^t \cdot rk_{i\rightarrow j} =  2{\bf{s}} ^t\left [ \begin{array}{c  |  c  | c} { \tilde{\A}}_{j} & {\A}_{j1} & {\A}_{j2}  \end{array} \right ]+   \tilde{\e} ^t + ( \mathbf{0},  \mathbf{0}, encode ({\bf{m}}) ^t)$, where $ \tilde{\e} =  (\tilde{\e}_{0}, \tilde{\e}_{1},\tilde{\e}_{2})$, $\tilde{\e}_{0} = ({\tilde{\bar{\e}}}_{0} , {\tilde{\e}}_{0}') $ and ${\tilde{\bar{\e}}}_{0} = {\bar{\e}}_{0} $, $ { \tilde {\e}}'_{0} = \bar{\e}_{0} {\X}_{00} +{\e}'_{0} {\X}_{10}+ {\e}_{1} {\X}_{20}$, $\tilde {\e}_{1} =  \bar{\e}_{0} {\X}_{01} +{\e}'_{0} {\X}_{11}+ {\e}_{1} {\X}_{21}$, 
$\tilde {\e}_{2} =  \bar{\e}_{0} {\X}_{02} +{\e}'_{0} {\X}_{12}+ {\e}_{1} {\X}_{22} + {\e}_2$.\\
Then output $ ct' = {\b}'$.

\subsection{Correctness and Security}
In this section, we analyze the correctness and security of the proposed scheme. 
\begin{theorem} [Correctness] 
\label{correctness}
The $\IBPRE$ scheme with parameters proposed in Section~\ref{sec:construction}  is correct.
\end{theorem}
\begin{proof}
To show that the decryption algorithm outputs a correct plaintext, it is required to consider both original and re-encrypted ciphertext. The arguments for the original ciphertext follows from the Lemma 6.2 of \cite{MP12}. For re-encrypted ciphertext, the main point is to consider the growth of error due to re-encryption. Argument for the controlled growth of error of  re-encrypted ciphertext follows, with some modifications, from Lemma 15 of \cite{KIRSH14}. Details calculations are omitted due to space constrained.
\end{proof}

\begin{theorem}[Security]
\label{thm:security}
The above scheme is $\INDsCPA$ secure assuming the hardness of decision-{\em LWE}$_{q, \alpha '}$ for $ \alpha' =\alpha /3 \geq 2 \sqrt{n} /q.$
\end{theorem}
\begin{proof}
First, using the same technique in \cite{MP12}, we transform the samples from LWE distribution to what we will need below. Given access to an LWE distribution ${\A}_{s,\alpha '}$ over $\mathbb{Z}_q^{n} \times \mathbb{T}$, (where$ \mathbb{T =R/Z}$) for any ${\bf {s}} \in \mathbb{Z}_q^{n} $, we can transform its samples $({\bf{a}}, b=\langle {\bf{s}},{\bf{a}}\rangle / q + e \mod 1)$ to have the form $({\bf{a}}, 2(\langle {\bf{s}},{\bf{a}}\rangle \mod q )+ e ' \mod 2q)$ for $ e' \leftarrow D _{\mathbb{Z},\alpha q}$, by mapping $b \mapsto 2qb + D_{\mathbb{Z}-2qb,s} \mod 2q$, where $s^{2} = (\alpha q)^{2}-(2 \alpha ' q)^{2} \geq 4n \geq \eta _\epsilon (\mathbb{Z}) ^{2}$, $\eta _\epsilon$ is smoothing parameter \cite{MR04,MP12}.  This transformation maps the uniform distribution over $\mathbb{Z}_q^{n} \times \mathbb{T} $ to the uniform distribution $\mathbb{Z}_q^{n} \times \mathbb{Z}_{2q}$. Once the LWE samples are of the desired form, we construct column-wise matrix $\A^{*}$ from these samples ${\bf{a}}$ and a vector ${\bf{b}^{*}}$ from the corresponding b. Let $id_{i^*}$ be the target user. The proof proceeds in a sequence of games.

\noindent{\bf Game 0:} This is the original $\INDsCPA$ game from definition between an attacker  $\mathcal{A}$ against scheme and an $\INDsCPA$ challenger.

\noindent{\bf Game 1:} In $\bf Game 1$ we change the way that the challenger generates $\bar{\A}, \bar {\A}', {\A}_1, {\A}_2$ in the public parameters. In $\bf{SetUp}$ phase, do as follows:
	\begin{itemize}
	 \item Set the public parameter $ \bar {\A} = \A^{*}$, where $\A^{*}$ is from LWE instance $(\A^{*}, {\b}^*)$ and set $\bar{\A}' =  -{\A}^{*}{\R} - {\H}_{id_{i^*}}\G$, where $\R$ is chosen in the same way as in ${\bf Game~0}$. 
	 \item Choose four invertible matrices ${\H}_1, {\H}_2, {\H}_3, {\H}_4$ at uniformly random from $\mathbb{Z}_q^{n\times n}$. 
	 \item Choose ${\R}_{{i^*}1},{\R}_{{i^*}2 }\leftarrow \mathcal{D}= D_{\mathbb{Z}, r}^{m \times nk }$; Set ${\A}'_1 = - \left [ \begin{array}{c  |  r} {\A}^{*}  & -{\A}^{*}{\R} \end{array} \right ] \cdot {\R}_{{i^*}1}$ and \\${\A}'_2 =  -\left [ \begin{array}{c  |  r} {\A}^{*}  & -{\A}^{*}{\R} \end{array} \right ] \cdot {\R}_{{i^*}2}$; Construct ${\A}_1 = {\A}'_1 -  {\H}_3 {\H}_{id_{i^*}}{\G}$ and ${\A}_2 = {\A}'_2 -  {\H}_4{\H}_{id_{i^*}}{\G}$.
\item Set $PP = (\bar{\A}, \bar {\A}', {\A}_1, {\A}_2, {\H}_1, {\H}_2, {\H}_3, {\H}_4,  \G)$ and send it to the Adversary $\mathcal{A}$.
\end{itemize}	
	To answer a secret key query against $id_i\in CU$, challenger will do the following: Construct	${\tilde{\A}}_{i} =
	\left [ \begin{array}{c|r} 
	{\A}^*  &  -{\A}^* {\R} -{\H}_{id_{i^*}}{\G}+ {\H}_{id{_i}}{\G} \end{array} \right ]  
	= \left [ \begin{array}{c  |  r}{ \A}^* &  -{\A}^* {\R} +({\H}_{id{_i}}-{\H}_{id_{i^*}})\G
\end{array} \right ].$
So, $\R$ is a trapdoor of ${\tilde \A_{i}}$ with invertible tag $({\H}_{id{_i}}-{\H}_{id_{i^*}})$. Then using $\mathbf{Extract}$ algorithm, challenger gets the secret key $sk_{id_{i}}= \left [ \begin{array}{c  |  c  }{ \R}_{i1} & { \R}_{i2} \end{array} \right ]$ for $id_i$, sends $sk_{id_{i}}$ to the adversary $\mathcal{A}$.\\ 
Challenger will send $\bot$, against the secret key query for $id_i\in HU$.

Note that  for $id_{i^*}$, $\tilde{\A}_{i^*} = \left [ \begin{array}{c  |  r} {\A}^*  &  -{\A}^* {\R}\end{array} \right ]$, so ${\A}'_1 =-{\tilde{\A}}_{i^*}{\R}_{{i^*}1}, {\A}'_2 = -{ \tilde{\A}}_{i^*} {\R}_{{i^*}2}$ and \\${\A}_{i^*} = \left [ \begin{array}{c  |  c  | c} { \tilde{\A}}_{i^*} & {\A}_{{i^*}1} & {\A}_{{i^*}2}  \end{array} \right ] = \left [ \begin{array}{c  |  c  | c} { \tilde{\A}}_{i^*} & {\A}_1 + {\H}_3 {\H}_{id_{i^*}}{\G} & {\A}_2 +  {\H}_4{\H}_{id_{i^*}}{\G}\end{array} \right ] = \left [ \begin{array}{c  |  c  | c} { \tilde{\A}}_{i^*} & {\A}'_1 & {\A}'_2 \end{array} \right ] =  \\ \left [ \begin{array}{c  |  c  | c} { \tilde{\A}}_{i^*} & - {\tilde{\A}}_{i^*}{\R}_{{i^*}1} & -{ \tilde{\A}}_{i^*} {\R}_{{i^*}2}\end{array} \right ]$

For the re-encryption key query and re-encryption query, challenger maintain the restrictions as in definition \ref{def:security} and computes $rk_{i \rightarrow j}$, $\mathbf{ReEnc}(rk_{i \rightarrow j}, ct)$ according to the $\mathbf{ReKeyGen}$ and $\mathbf{ReEnc}$ algorithms to reply the adversary. 
Due to left-over hash lemma \cite[Lemma 14]{ABB10-EuroCrypt}, $({\A}^{*}, -{\A}^{*}{\R}, -\left [ \begin{array}{c  |  r} {\A}^{*}  & -{\A}^{*}{\R} \end{array} \right ] \cdot {\R}_{{i^*}1}, -\left [ \begin{array}{c  |  r} {\A}^{*}  & -{\A}^{*}{\R} \end{array} \right ] \cdot {\R}_{{i^*}2})$ is statistically indistinguishable with uniform distribution. Hence, $({\A}^{*}, -{\A}^{*}{\R}- {\H}_{id_{i^*}}{\G}, - \left [ \begin{array}{c  |  r} {\A}^{*}  & -{\A}^{*}{\R} \end{array} \right ] \cdot {\R}_{{i^*}1}- {\H}_3{\H}_{id_{i^*}}{\G}, -\left [ \begin{array}{c  |  r} {\A}^{*}  & -{\A}^{*}{\R} \end{array} \right ] \cdot {\R}_{{i^*}2}-{\H}_4{\H}_{id_{i^*}}{\G})$ is statistically indistinguishable with uniform distribution. Since $\bar{\A}, \bar {\A}', {\A}_1, {\A}_2$ and responses to key queries are statistically close to those in ${\bf Game~0}$, ${\bf Game ~0}$ and ${\bf Game ~1}$ are statistically indistinguishable.

\noindent{\bf Game 2:} In $\bf Game 2$ we change the way that the challenger generates challenge ciphertext. Here Challenger will produce the challenge ciphertext $\b$ on a message ${\bf{m}}\in \{0,1\}^{nk}$ for $id_{i^*}$ as follows: Choose ${\bf{s}} \leftarrow \mathbb{Z}_q^{n}$ and $\bar{\e} _0 \leftarrow D_{\mathbb{Z},\alpha q}^{\bar{m}}$ as usual, but do not choose ${\e}_0', {\e}_1, {\e}_2$.
Let ${\bar{\b}_{0}}^t = 2({\bf{s}}^t {\A ^{*}} \mod q)+ \bar{\e}_0^t  \mod 2q$ and ${{\b}_{0}'}^t = -{\bar{\b}_{0}}^t {\R} + {\hat \e}_{0} ^t \mod 2q$, where ${\hat\e}_0 \leftarrow D_{\mathbb{Z},s'}^{nk}$. So, ${\b}_0 = ({\bar{\b}_{0}}, {{\b}_{0}'})$. The last $2nk$ coordinates can be set as ${\b}_{1}^t =  -{\b}_0 ^t {\R}_{{{i}^*} 1}+ {\hat \e}_{1}^t \mod 2q$;
${\b}_{2}^t = -{\b}_0 ^t {\R}_{{{i}^*} 2}+ {\hat \e}_{2}^t+ encode ({\bf{m}})\mod 2q$, where ${\hat \e}_1, {\hat \e}_2 \leftarrow D_{\mathbb{Z},s'}^{nk}$. Finally, replace ${\bar{\b}_{0}}$ with $\b^{*}$ in all the above expression, where $({\A }^{*}, \b^{*})$ is the LWE instance. Therefore, ${\bar{{\b}_{0}}}^t = {{\b}^{*}}^t$;
${{\b}_{0}'}^t =- {{\b}^{*t}} {\R} + {\hat \e}_{0} ^t \mod 2q$;
${\b}_{1}^t =  -{{\b}_{0}^{*} }^t {\R}_{{{i}^*} 1}+ {\hat \e}_{1}^t \mod 2q$;
${\b}_{2}^t = -{{\b}_{0}^{*}}^t {\R}_{{{i}^*} 2}+ {\hat \e}_{2}^t+ encode ({\bf{m}})\mod 2q$. Set  ${{\b}_{0}^{*}}^t = ({{\b}^{*}}^t,-{{\b}^{{*}t}} {\R} + {\hat\e}_{0}^t \mod 2q)$. Then the challenger output the challenge ciphertext $ct= {\b} = ({{\b}_{0}^{*}}, {\b}_1, {\b}_2)$.

We now show that the distribution of $\b$ is within $\negl$ statistical distance of that in ${\bf Game ~1}$ from the adversary's view. Clearly, ${\b}^{*}$  have essentially the same distribution as in ${\bf Game ~0}$ by construction.
By substitution we have: 
${{\b}_{0}'}^t  = 2({\bf{s}}^t (-{\A ^{*}} {\R}) \mod q)+  \bar{\e}_0^t {\R} +  {\hat \e}_{0}^t \mod 2q$;
${\b}_{1}^t   = 2({\bf{s}}^t(-{\tilde{\A}}_{i^{*}}{\R}_{{{i}^*} 1}) \mod q) + (  \bar{\e}_0^t ,  \bar{\e}_0^t {\R}+ \hat{\e}_{0}^t) {\R}_{{{i}^*} 1} + {\hat\e}_{1}^t \mod 2q$;
${\b}_{2}^t =2({\bf{s}}^t(-{\tilde{\A}}_{i^{*}}{\R}_{{{i}^*} 2}) \mod q)+ (  \bar{\e}_0^t ,  \bar{\e}_0^t {\R}+ \hat{\e}_{0}^t) {\R}_{{{i}^*} 2}+ \hat{\e}_{2}^t +  encode ({\bf{m}}) \mod 2q.$

By Corollary 3.10 in \cite{Regev05}, the noise term $\bar{\e}_0^t {\R}+ \hat{\e}_{0}^t$ of ${\b}_{0}' $ is within $ \negl$ statistical distance from discrete Gaussian distribution $D_{\mathbb{Z},s'}^{nk}$. The same argument, also, applies for the noise term of ${\b}_1, {\b}_2$. Hence, ${\bf Game ~1}$ and ${\bf Game ~2}$ are statistically indistinguishable. 

\noindent{\bf Game 3:} Here, we only change how the ${\b}^{*}$ component of the challenge ciphertext is created, letting it be uniformly random in $\mathbb{Z}_{2q} ^{\bar m}$. Challenger construct the public parameters, answer the secret key queries, re-encryption queries and construct the last $3nk$ coordinates of challenge ciphertext exactly as in Game~2. It follows from the hardness of the decisional LWE$_{q, \alpha'}$ that ${\bf Game ~2}$ and ${\bf Game~3}$ are computationally indistinguishable.

Now, by the left-over hash lemma \cite[Lemma 14]{ABB10-EuroCrypt}, (${ \A}^*,{\b}^{*}, -{ \A}^{*}{\R},{{\b}^{*t}} {\R}, -\tilde{\A}_{i^{*}} {\R }_{{i^*}1},\\{{\b}_{0}^{*} }^t {\R}_{{{i}^*} 1},$  $-\tilde{\A}_{i^{*}} {\R }_{{i^*}2},$ ${{\b}_{0}^{*} }^t {\R}_{{{i}^*} 2}$) is $\negl$-uniform when ${ \R},  {\R}_{{{i}^*} 1}, {\R}_{{{i}^*} 2}$ are chosen as in Game 2. Therefore, the challenge ciphertext has the same distribution (up to  $\negl$ statistical distance) for any encrypted message. So, the advantage of the adversary against the proposed scheme is same as the advantage of the attacker against decisional LWE$_{q, \alpha'}$. 
\end{proof}

\bibliography{latbib}
\bibliographystyle{splncs04}

\end{document}